\documentclass[12pt,reqno, letterpaper]{amsart}
\usepackage{amssymb,epsf}
\usepackage{amstext,amsmath,amsfonts,amscd,amsthm,graphicx}
\usepackage{epsfig}
\usepackage{eufrak}
\usepackage{latexsym}
\usepackage{eucal}
\theoremstyle{amsart}
\newfont{\fnt}{cmsy10}
\newfont{\sss}{cmr10}
\newfont{\azb}{wncyr10}
\newfont{\azbit}{wncyi10}
\theoremstyle{definition}

\theoremstyle{plain}

\newtheorem{vt}{Theorem}
\newtheorem{lm}{Lemma}

\newtheorem{prp}{Proposition}
\theoremstyle{definition}
\newtheorem{pz}{Remark}
\newtheorem{pr}{Example}

\begin{document}
\title[Low-order Hamiltonian operators having momentum]{
{\protect\vspace*{-1cm}}Low-order Hamiltonian operators having momentum}
\author{Ji\v{r}ina Vodov\'a}
\address{Mathematical Institute, Silesian University in Opava, Na Rybn\'{i}\v{c}ku \nolinebreak 1, 746 01 Opava, Czech Republic}
\email{Jirina.Vodova@math.slu.cz}
\maketitle
\begin{abstract} {\protect\vspace*{-0.7cm}}
We describe all fifth-order Hamiltonian operators in one dependent and one independent variable that possess momentum, i.e., for which there exists a Hamiltonian associated with translation in the independent variable. Similar results for first- and third-order Hamiltonian operators were obtained earlier by Mokhov.
\looseness=-1
\end{abstract}
\section{Introduction}
Hamiltonian evolution equations are well known to play an important role in modern
mathematical physics \cite{kac,dorfman,krasilshchik,mokhov2}. Indeed, a Hamiltonian operator maps the variational derivatives of
the conserved quantities into symmetries; this is of particular significance
in the theory of integrable systems which often turn out to be \textit{bi}-Hamiltonian,
see e.g.\ \cite{dickey, dorfman, fokas, fokas2,olver,olver3, sergyeyev, sergyeyev3} and references therein.

In the present paper we employ the so-called special contact transformations, first introduced in \cite{mokhov1985} and considered later in more detail  in \cite{mokhov}, to classify fifth-order Hamiltonian operators admitting momentum, see below for details.
Special contact transformations preserve existence of momentum, see \cite{mokhov1985,mokhov}, and for this reason
we study the existence of momentum for just the representatives of the associated equivalence classes.  \looseness=-1

Existence of momentum is useful for averaging the corresponding Hamiltonian systems, see e.g. \cite{dubrovin2}.
Hamiltonian operators having momentum could be employed e.g.\ for the generation of hierarchies
of local symmetries (i.e., higher commuting flows) in the following fashion.

Suppose we are given a nonzero Hamiltonian operator, say $\mathfrak{D}$, 
in one dependent variable $u$ and one independent variable $x$ possessing momentum, i.e., there exists a functional
$\mathcal{P}=\int h dx$ such that $u_x=\mathfrak{D}\delta_u \mathcal{P}$.
Further assume that there exists another translation-invariant Hamiltonian operator $\mathfrak{E}$ which is compatible with $\mathfrak{D}$ and such that the operator $\mathfrak{R}=\mathfrak{E}\circ\mathfrak{D}^{-1}$ is a weakly nonlocal hereditary operator. Then $\mathfrak{R}$ is a recursion operator for the equation $u_{t_0}=u_x$, and under a further minor technical assumption of normality of $\mathfrak{R}$ in the sense of \cite{sergyeyev2}, by Theorem~1 from  \cite{sergyeyev2} the quantities $\mathfrak{R}^i(u_x)$ are local and the associated flows commute for all $i=1,2,3,\dots$, i.e., we have an infinite hierarchy of local commuting flows $u_{t_j}=\mathfrak{R}^j(u_x)$, $j=0,1,2,\dots$.

\section{Preliminaries}

In what follows we are going to deal with Hamiltonian operators and associated Hamiltonian evolution
equations involving a single spatial variable $x$ and a single dependent variable $u$.
A Hamiltonian evolution equation takes the form
$$u_t=\mathfrak{D}\delta_u\mathcal{T}[u],$$
where $\mathfrak{D}$ is a Hamiltonian operator,
$\mathcal{T}=\int T[u]dx$ is a functional (often referred to as the Hamiltonian),
$\delta_u$ denotes the
variational derivative with respect to $u$, and the notation $T[u]$ 
indicates that $T$ is a differential function, see the definition below.
Recall (see e.g.\ \cite{olver} for details) that a Hamiltonian operator $\mathfrak{D}$ defines the Poisson bracket
$$\left\{\mathcal{R},\mathcal{S}\right\}=\int\delta_u\mathcal{R}\mathfrak{D}\delta_u\mathcal{S}
\mathrm{d}x$$ which should satisfy certain conditions, namely,
\textit{skew symmetry}$$\left\{\mathcal{R},\mathcal{S}\right\}
=-\left\{\mathcal{S},\mathcal{R}\right\}$$
and the \textit{Jacobi identity}
$$\left\{\left\{\mathcal{R},\mathcal{S}\right\},\mathcal{T}\right\}+
\left\{\left\{\mathcal{S},\mathcal{T}\right\},\mathcal{R}\right\}+
\left\{\left\{\mathcal{T},\mathcal{R}\right\},\mathcal{S}\right\}=0$$
which must hold for all admissible functionals $\mathcal{R},\mathcal{S},$ and $\mathcal{T}$.

It can be shown that the
skew-symmetry condition is equivalent to the skew-adjointness of the operator $\mathfrak{D}$.
Let $D_x$ denote the total derivative with respect to the spatial variable $x$ and $u_i\equiv D_x^i(u)$.
Recall (see e.g.\ \cite{olver}) that a \textit{differential function} by definition depends on 
$x$, $u$, and finitely many derivatives $u_j$ of $u$
with respect to the space variable $x$.

In \cite{dorfman}, for any operator $\mathfrak{D}=\sum_{k=0}^N p_k D_x^k$
and for any differential function $f$ the author defines another
differential operator $D_\mathfrak{D}f$ by the formula
$$(D_\mathfrak{D}f)h=\left(\mathrm{pr\ v}_h(\mathfrak{D})\right)(f),$$
where $\mathrm{pr\ v}_h$ is the prolongation of a vector field $\mathrm{v}_h$ with the
characteristic $h$, i.e.,
$$
\mathrm{pr\ v}_h=h\frac{\partial }{\partial u}+\sum_{i=1}D^i(h)\frac{\partial}{\partial u_i}  .
$$

We have
$$D_\mathfrak{D}f=\sum_{k,m}\frac{\partial p_k}{\partial u_{m}}D^k (f) D^m,$$ and it can be shown \cite{dorfman} that
the Jacobi identity for the skew-adjoint operator $\mathfrak{D}$ is equivalent to the condition
\begin{equation}\label{dorfman}(D_\mathfrak{D}h_1)\mathfrak{D}h_2-(D_\mathfrak{D}h_2)\mathfrak{D}h_1
+\mathfrak{D}(D_\mathfrak{D}h_1)^*h_2=0\end{equation}
which must hold for arbitrary smooth differential functions $h_1$ and $h_2$.
\looseness=-1

Recall (see e.g.\ \cite{kac}) that the \textit{differential order} of a differential function $f$,
denoted by $\mathrm{ord}(f)$, is the maximal $m\in\mathbb{Z}_+$
such that $\frac{\partial f}{\partial u_m} \neq 0$ if $f$ is not a quasiconstant ($f$ is \textit{quasiconstant} if it depends only on the spatial variable $x$), and is $-\infty$
if $f$ is quasiconstant. Following \cite{dorfman} define the \textit{level} $m$ of the Hamiltonian operator
$\mathfrak{D}=\sum_{k=0}^N p_k D^k$  of order $N$ to be
$m=\max\limits_{j}\left\{j+\mathrm{ord}(p_j)\right\}$.
The possible values of the level of a nonquasiconstant-coefficient Hamiltonian operator were studied e.g.\
in \cite{dorfman, kac}. In this paper we are specifically interested in fifth-order nonquasiconstant-coefficient Hamiltonian operators, whose only
possible level values $m$ are $m=5,6$ or $7$ \cite{dorfman}.

Following \cite{mokhov} we say that a Hamiltonian operator $\mathfrak{D}$ \textit{has momentum}
 if there exists a functional $\mathcal{T}$, referred to as a {\em momentum}, such that
$$\mathfrak{D}\delta_u\mathcal{T}=u_1.$$

Differential substitutions are among the most important tools using which we can distinguish 
Hamiltonian operators that have momentum from those that have not. The following lemma shows
how the Hamiltonian operators behave under differential substitutions:
\begin{lm}[\cite{mokhov}]\label{lm1} Let $\mathfrak{D}_1$ be a Hamiltonian operator in the variables $x, u$.
Under the transformation
\begin{equation}\label{1}x=\varphi(y,v,v_1,\dots,v_{m}),\quad u=\psi(y,v,v_1,\dots,v_{n}),\end{equation}
where $v_{j}=D_y^j(v)$, and $D_y$ is the total derivative with respect to $y$,
the operator $\mathfrak{D}_1$ goes into the Hamiltonian operator $\mathfrak{D}_2$ defined by the formula
\begin{equation}\overline{\mathfrak{D}}_1=(D_y(\varphi))^{-1}K^*\circ \mathfrak{D}_2\circ K,\end{equation}
where $$K=\sum_{i=0}^{\max(m,n)}(-1)^iD_y^i\circ\left(\frac{\partial\psi}{\partial v_{i}}D_y(\varphi)
-\frac{\partial\varphi}
{\partial v_{i}}D_y(\psi)\right),$$
 $K^*$ is the formal adjoint of $K$, and
$\overline{\mathfrak{D}}_1$ is obtained from $\mathfrak{D}_1$ upon using (\ref{1}) and
setting $D_x=(D_y(\varphi))^{-1}D_y$.
\end{lm}
\begin{pz}\label{pz1}
 A differential substitution preserves the order of a scalar local Hamiltonian operator if and only if this substitution is a contact transformation, see Theorem 1 in \cite{mokhov}.
Note that in general the operator $\mathfrak{D}_2$ may contain nonlocal terms unless (\ref{1}) is a contact
transformation, cf.\ e.g.\ \cite{astashov, kac, mokhov}.
\end{pz}

General contact transformations  do not preserve the property of having momentum. However, in \cite{mokhov1985}
Mokhov introduced a pseudogroup of
\textit{special contact transformations}
$$x=\varphi(y,v,v_y)=y+w(v,v_y),u=\psi(v,v_y),$$
$$\frac{\partial\varphi}{\partial v_y}D_y(\psi)=\frac{\partial\psi}{\partial v_y}D_y(\varphi),
\rho=\frac{\partial\psi}{\partial v}-\frac{\partial\varphi}{\partial v}D_y(\psi)/D_y(\varphi)\not\equiv 0$$
which preserve existence of momentum, see \cite{mokhov1985, mokhov}.

\begin{pz}
Note that  special contact transformations play an important role also in other applications, e.g. in  the theory of canonical variables for the two-dimensional hydrodynamics of an incompressible fluid with vorticity \cite{mokhov1989, mokhov1989_2}.
\end{pz}
It is readily checked that the coefficients of Hamiltonian operators having momentum may not explicitly depend on the spatial variable $x$, that is, the Hamiltonian operators having momentum must be translation-invariant. This happens because in order to possess momentum
the Lie derivative of the Hamiltonian operator in question along the vector field with the characteristic $u_1$ must vanish.

\section{First- and third-order Hamiltonian operators having momentum}

Mokhov \cite{mokhov} has shown that for the first-order Hamiltonian operators the condition of translation invariance is not only necessary but also sufficient for the existence of momentum:
\begin{prp}[\cite{mokhov}]
A first-order Hamiltonian operator has momentum if and only if it is translation-invariant.
\end{prp}

Classification of the third-order translation-invariant operators
under a special contact transformation was also obtained by Mokhov. He employed it to find out whether a given third-order translation-invariant Hamiltonian operator has momentum:

\begin{prp}[\cite{mokhov}] An arbitrary translation-invariant Hamiltonian operator of
the third order can be reduced by a special contact transformations to one of the operators (\ref{OP1})-(\ref{OP3}).
\begin{enumerate}
\item{An operator
\begin{equation}\label{OP1}
\mathfrak{D}=\pm \frac{1}{u_x}\left[D_x^3+2SD_x+D_xS\right]\circ\frac{1}{u_x}+
2fD_x+D_xf,
\end{equation}
where  $S=\frac{u_3}{u_1}-\frac{3}{2}\frac{(u_2)^2}{(u_1)^2}$  and $f$ is an arbitrary function of $u$ only, has momentum. The corresponding functional is of the form  $\int p(u)\mathrm{d}x$,
 where $p(u)$ is the solution of the equation
$$\pm\frac{\partial^4 p}{\partial u^4}+2f(u)\frac{\partial^2 p}{\partial u^2}
+\frac{\partial p}{\partial u}\frac{\partial f}{\partial u}-1=0.$$}
\item{An operator \begin{equation}\label{OP2}
\mathfrak{D}=\pm \left[D_x^3+2AuD_x+Au_x\right], A=\mathrm{const}>0
\end{equation}
 has momentum.}
\item{An operator \begin{equation}\label{OP3}
\mathfrak{D}=\pm \left[D_x^3+AD_x\right], A=\mathrm{const}.
\end{equation}
 does not have momentum.}

\end{enumerate}
\end{prp}
\section{Fifth-order Hamiltonian operators having momentum}
As proved above, one should look for Hamiltonian operators having momentum
among the translation-invariant ones. Below we will classify fifth-order translation-invariant
Hamiltonian operators according to their leading coefficients up to special contact transformations 
which preserve the property of having momentum. Our first result in this direction is as follows:
\begin{prp}
Any fifth-order translation-invariant Hamiltonian operator can be reduced by a special contact transformation
to an operator with leading coefficient equal to either $\pm 1$ or  $\pm \frac{1}{u_1^4}$.
\end{prp}
\begin{proof}
The proof partially uses the line of reasoning analogous to the one used by Mokhov in his classification of third-order operators. The leading coefficient of a fifth-order translation-invariant Hamiltonian operator in a single spatial variable $y$ and a single dependent variable $v$
has the
general form  (see \cite{cooke})
$$\pm\frac{1}{(\alpha v_2+\beta)^6},\ \alpha=\alpha(v,v_1),\ \beta=\beta(v,v_1),$$
where $v_j=D_y^j (v)$.
If $\alpha\not\equiv 0$, then we can find a special contact
transformation to get rid of the dependence of the leading coefficient of the operator on $v_{2}$ in
the following way:
take a function $\tilde{w}(v,v_1)$ such that
$\frac{\partial \tilde{w}}{\partial v}v_1+1
=\frac{\beta}{\alpha}\frac{\partial\tilde{w}}{\partial v_1}$, $\frac{\partial\tilde{w}}{\partial v_1}
\not\equiv 0$, and  a function $\tilde{\psi}(v,v_1)\not\equiv 0$ such that
$(1+D_y(\tilde{w}))\frac{\partial\tilde{\psi}}{\partial v_1}=\frac{\partial\tilde{w}}
{\partial v_1}D_y(\tilde{\psi})$ and  $\tilde{\rho}
:=\frac{\partial\tilde{\psi}}{\partial v}-\frac{\partial\tilde{w}}{\partial v}
\frac{D_y(\tilde{\psi})}{(1+D_y(\tilde{w}))}\not\equiv 0.$ The functions $\tilde{w}$ and $\tilde{\psi}$
define  a special contact transformation
\begin{equation}\label{clas1}x=y+\tilde{w}(v,v_1),\ u=\tilde{\psi}(v,v_1),\end{equation}
where $x$ is a new independent variable and $u$ is a new dependent variable.
The inverse of (\ref{clas1}) is also a contact transformation:
\begin{equation}\label{clas2}y=x+w(u,u_1)=\varphi(x,u,u_1),\ v=\psi(u,u_1),\end{equation} and it can be  verified that the
leading coefficient of the operator  transformed by (\ref{clas2}) does not depend on $u_{2}$.

Now suppose that $\alpha\equiv 0$. Then \cite{cooke} the leading coefficient of our operator
is of the form $\pm\frac{1}{(\tilde{\alpha}v_1+\tilde{\beta})^4},\ \tilde{\alpha}=\tilde{\alpha}(v),\
\tilde{\beta}=\tilde{\beta}(v)$. If $\tilde{\beta}\equiv 0$,  it is impossible to
get rid of the dependence of the leading coefficient of the operator on $v_1$ using special
contact transformations alone. Using the transformation $$y=x+w(u),\ v=\psi(u),$$ 
where $\psi(u)$ is such that
$\frac{\partial \psi}{\partial u}=\sqrt[3]{1/\tilde{\alpha}^2}$, makes the leading coefficient equal to $\pm\frac{1}{u_1^4}$. If $\tilde{\beta}\not\equiv 0$, a special
contact transformation $$y=x+w(u), \ v=\psi(u),$$ which is an inverse of the transformation
$$x=y+\tilde{w}(v),u=\tilde{\psi}(v)\not\equiv\mathrm{const},
 \frac{\partial \tilde{w}}{\partial v}=\frac{\tilde{\alpha}}{\tilde{\beta}},$$
turns our operator into an operator with a leading coefficient that does not depend on $u_1$.

If the differential order of the leading coefficient is equal to zero
(i.e., the leading coefficient of our translation-independent operator depends only on $v$,
and is therefore of the form $\frac{1}{\alpha(v)}$),  the special contact transformation
$$y=x, v=\psi(u), \left(\frac{\partial\psi}{\partial u}\right)^2=\pm\frac{1}{\alpha(v(u))}$$ makes the leading coefficient of our transformed operator equal to $\pm 1$.
\end{proof}

In what follows a fifth-order Hamiltonian operator is supposed to be written in the form
$$\mathfrak{D}=aD_x^5+D_x^5\circ a+bD_x^3+D_x^3\circ b+cD_x+D_x\circ c,$$ which ensures  
skew-adjointness of the operator and hence skew-symmetry of the associated Poisson bracket.

Next we are going to show that no fifth-order Hamiltonian operator with the leading coefficient $\pm 1$ has momentum and find out what the momenta $\mathcal{P}$ for the operators with the leading coefficient $\pm 1/u_1^4$ look like. Note that the operators with the leading coefficient $-1$  can be transformed by the special contact transformation  $x=y$, $u=iv$, where $i=\sqrt{-1}$ is the imaginary unit, to operators with the leading coefficient $1$ so
it is sufficient to show that no fifth-order Hamiltonian operator with the leading coefficient $1$ has momentum. The same reasoning could be applied to operators with the leading coefficient $-1/u_1^4$, but as we are interested in  finding the explicit form of the momentum functionals $\mathcal{P}$ for operators with the leading coefficients $1/u_1^4$ and  $-1/u_1^4$, we discuss each of these cases separately.

Thus,  we first try to find general forms of fifth-order translation-invariant Hamiltonian operators with the leading coefficients $1$, $1/u_1^4$ and $-1/u_1^4$.
For the case of Hamiltonian operators with the leading coefficient equal to
$1$ this was done (even though in a more general setting than we actually need) in \cite{cooke}:\looseness=-1
\begin{lm}[\cite{cooke}]
A fifth-order Hamiltonian operator whose leading coefficient is $1$
 must be of the form
$$\mathfrak{D}=D_x^5+bD_x^3+D_x^3\circ b+cD_x+D_x\circ c,$$
where $b$ and $c$ are functions of $x$ alone or they are given by the formulas
\begin{eqnarray*}
b&=&\frac{3}{2}(u+\alpha)^{-1}(u_{xx}+\alpha^{\prime\prime})
-\frac{7}{4}(u+\alpha)^{-2}(u_x+\alpha^{\prime})^2+\beta(u+\alpha)+\gamma,\\
c&=&-\frac{z_4}{z}+\frac{\beta z_1^2}{2z}+\frac{wz_2}{2z}-\frac{wz_1^2}{4z^2}-\frac{w_1z_1}{z}
+\frac{9z_1z_3}{2z^2}-\frac{129z_1^2 z_2}{8z^3}+\frac{273z_1^4}{32z^4}\\&&+\frac{33 z_2^2}{8z^2}
-\frac{\beta z_2}{2}-\frac{3z\beta^{\prime\prime}}{2}-\frac{\beta^{\prime}z_1}{2}
-\frac{\beta^{2}z^2}{2}+\frac{w^2}{2},
\end{eqnarray*}
where $\alpha$, $\beta$, and $\gamma$ are functions of $x$ only, $w$ and $z$ are
given by
$$w=\beta z+\gamma,\ z=u+\alpha,$$
and $w_i=D_x^i(w)$, $z_i=D_x^i(z)$.

If  $\beta=0$, then any choice of $\alpha$ and $\gamma$
yields a Hamiltonian operator.

If $\beta\neq 0$,
then
$$\gamma=-\frac{\rho}{\beta^2}-\frac{\beta^{\prime\prime}}{2\beta}
+\frac{(\beta^{\prime})^2}{4\beta^2},$$
where $\rho$ is an arbitrary constant.
\end{lm}

On the other hand, using the results of \cite{cooke} we can now prove the following
\begin{lm}\label{leadcoeff2}
A fifth-order translation-invariant Hamiltonian operator whose leading coefficient is $\pm 1/u_1^4$ must be of the form
$$\mathfrak{D}=\pm\frac{1}{2u_1^4 }D_x^5\pm D_x^5\circ\frac{1}{2u_1^4 }+bD_x^3+D_x^3\circ b+
cD_x+D_x\circ c,$$
where
\begin{eqnarray*}
b&=&\frac{1}{2u_1^6}\left(\pm 10u_3u_1\mp 55u_2^2+2\alpha u_1^4\right),\\
c&=&\frac{1}{u_1^8}\left(3u_1^6u_2\frac{\partial \alpha}{\partial u}+2u_1^5u_3\alpha
-6u_1^4u_2^2\alpha+\beta u_1^8\mp 3u_1^3u_5\pm 65u_1^2u_2u_4\pm 50u_1^2u_3^2\right.\\
&&\left.\mp 615u_1u_2^2u_3\pm 735u_2^4\right),
\end{eqnarray*}
and   $\alpha$ and $\beta$ are  functions of $u$ only. \end{lm}
\begin{pz}
It can be shown that there is no special contact transformation which preserves the leading coefficient and simultaneously eliminates one of the unknown functions $\alpha$, $\beta$.
\end{pz}
\begin{proof}
We will prove the lemma in question only for the case of the leading coefficient $1/u_1^4$. The proof for the case of the leading
coefficient equal to $-1/u_1^4$ can be obtained in a similar fashion.
Put $a=1/(2u_1^4)$. Then the Jacobi identity implies the following  relations  
(cf.\ \cite{cooke}):
{
\allowdisplaybreaks
\begin{eqnarray}
\frac{\partial c}{\partial u_6}&=&0\label{c6}\\
\frac{\partial b}{\partial u_4}&=&0\label{b4}\\
\frac{\partial c}{\partial u_5}&=&-\frac{3}{u_1^5}\label{c5}\\
\frac{\partial b}{\partial u_3}&=&\frac{5}{u_1^5}\label{b3}\\
\frac{\partial c}{\partial u_4}&=&\frac{1}{3u_1^6}\left(85u_2-2\frac{\partial b}{\partial u_2}u_1^6\right)\label{c4}\\
\frac{\partial c}{\partial u_3}&=&\frac{1}{3u_1^7}\left(-16\frac{\partial b}{\partial u_2}u_1^6u_2-225u_1u_3-9D_x\left(\frac{\partial b}{\partial u_2}\right)u_1^7+410u_2^2+6bu_1^6\right)\label{c3}\\
\frac{\partial b}{\partial u_1}&=&\frac{1}{3u_1^7}\left(26\frac{\partial b}{\partial u_2}u_1^6u_2+340u_1u_3+7D_x\left(\frac{\partial b}{\partial u_2}\right)u_1^7-550u_2^2-6bu_1^6\right)\label{b1}\\
\frac{\partial c}{\partial u_2}&=&\frac{1}{6u_1^8}\left(-3\frac{\partial b}{\partial u}u_1^8+140\frac{\partial b}{\partial u_2}u_1^6u_2^2+80bu_1^6u_2+11390u_1u_2u_3-96\frac{\partial b}{\partial u_2}u_1^7u_3-14260u_2^3\right.\nonumber\\
&&\left.-27D_x^2\left(\frac{\partial b}{\partial u_2}\right)u_1^8+21D_x(b)u_1^7-1200u_1^2u_4+2\frac{\partial b}{\partial u_2}bu_1^{12}-82D_x\left(\frac{\partial b}{\partial u_2}\right)u_1^7u_2\right)\label{c2}\\
\frac{\partial c}{\partial u_1}&=&\frac{1}{6u_1^9}\left(18\frac{\partial b}{\partial u}u_1^8u_2-42D_x\left(\frac{\partial b}{\partial u_2}\right)u_1^7u_2^2+416\frac{\partial b}{\partial u_2}u_1^6u_2^3-271D_x(b)u_1^7u_2-856bu_1^6u_2^2\right.\nonumber\\
&&\left.+80bu_1^7u_3+14730u_1^2u_2u_4-214D_x\left(\frac{\partial b}{\partial u_2}\right)u_1^8u_3-68\frac{\partial b}{\partial u_2}u_1^8u_4-136D_x^2\left(\frac{\partial b}{\partial u_2}\right)u_1^8u_2\right.\nonumber\\
&&\left.+2D_x\left(\frac{\partial b}{\partial u_2}\right)bu_1^{13}-4\frac{\partial b}{\partial u_2}D_x(b)u_1^{13}-92450u_1u_2^2u_3-1080u_1^{3}u_5-21D_x^3\left(\frac{\partial b}{\partial u_2}\right)u_1^9\right.\nonumber\\
&&\left.+3D_x^2(b)u_1^8+7610u_1^2u_3^2-3D_x\left(\frac{\partial b}{\partial u}\right)u_1^9-404\frac{\partial b}{\partial u_2}u_1^7u_2u_3-4\frac{\partial b}{\partial u_2}bu_1^{12}u_2+87920u_2^4\right)\label{c1}\\
\frac{\partial c}{\partial u}&=&\frac{1}{6u_1^{10}}\left(-21D_x^4\left(\frac{\partial b}
{\partial u_2}\right)u_1^{10}+9D_x^2\left(\frac{\partial b}{\partial u}\right)u_1^{10}
+28410u_1^3u_3u_4+15730u_1^3u_2u_5-66D_x^2(b)u_1^8u_2\right.\nonumber\\
&&\left.-608D_x(b)u_1^7u_2^2-717 D_x(b)u_1^8u_3-13280bu_1^6u_2^3-660bu_1^8u_4-205510u_1^2u_2u_3^2
-139340u_1^2u_2^2u_4\right.\nonumber\\
&&\left.+838900u_1u_2^3u_3-80D_x(b)\frac{\partial b}{\partial u_2}u_1^{13}u_2-1416
\frac{\partial b}{\partial u_2}u_1^7u_2^2u_3-2062D_x\left(\frac{\partial b}{\partial u_2}\right)
u_1^8u_2u_3\right.\nonumber\\
&&\left.-464\frac{\partial b}{\partial u_2}bu_1^{12}u_2^2-
32\frac{\partial b}{\partial u_2}bu_1^{13}u_3-120D_x
\left(\frac{\partial b}{\partial u_2}\right)bu_1^{13}u_2
-744\frac{\partial b}{\partial u_2}u_1^8u_2u_4-673120u_2^5\right.\nonumber\\
&&\left.+440cu_1^8u_2-6D_x(b)bu_1^{13}-232b^{2}u_1^{12}u_2-1092u_1^4u_6
-9D_x^3(b)u_1^9+6D_x(c)u_1^9\right.\nonumber\\
&&\left.-252D_x^3\left(\frac{\partial b}{\partial u_2}\right)
u_1^9u_2-956D_x^2\left(\frac{\partial b}{\partial u_2}\right)u_1^8u_2^2-408D_x^2\left(
\frac{\partial b}{\partial u_2}\right)u_1^9u_3+8c\frac{\partial b}{\partial u_2}u_1^{14}+6\frac{\partial b}{\partial u}b
u_1^{14}\right.\nonumber\\
&&\left.-4b^2\frac{\partial b}{\partial u_2}u_1^{18}
-1304D_x\left(\frac{\partial b}{\partial u_2}\right)u_1^7u_2^3
-282D_x\left(\frac{\partial b}{\partial u_2}\right)u_1^9u_4-2200\frac{\partial b}{\partial u_2}
u_1^6u_2^4-68\frac{\partial b}{\partial u_2}u_1^9u_5\right.\nonumber\\
&&\left.+66D_x\left(\frac{\partial b}{\partial u}\right)u_1^9u_2+192\frac{\partial b}{\partial u}u_1^8u_2^2
+12\frac{\partial b}{\partial u}u_1^9u_3
-12D_x^2\left(\frac{\partial b}{\partial u_2}\right)bu_1^{14}
-12D_x\left(\frac{\partial b}{\partial u_2}\right)D_x(b)u_1^{14}\right.\nonumber\\
&&\left.-708\frac{\partial b}{\partial u_2}
u_1^8u_3^2+3124bu_1^7u_2u_3
\right).\label{c0}
\end{eqnarray}
}

Now we can equate mixed partial derivatives of the coefficients $b$ and $c$ to obtain new relations. Equating
$\frac{\partial}{\partial u}\left(\frac{\partial c}{\partial u_4}\right)$ and
$\frac{\partial}{\partial u_4}\left(\frac{\partial c}{\partial u}\right)$,
evaluating the total derivatives and substituting for the partial derivatives of $c$ with respect
to $u_6,u_5,u_4$, $u_3,u_2$ and for the partial derivative of $b$ with respect to $u_4$ and $u_3$
from the relations (\ref{c6})--(\ref{c3}) and (\ref{c2}) gives\looseness=-1
\begin{eqnarray}\label{mixed}
0&=&1197u_1^8\frac{\partial^2b}{\partial u_2\partial u}+252u_1^9
\frac{\partial^3b}{\partial u_2\partial u_1\partial u}+2034bu_1^6
+378u_1^{10}\frac{\partial^4b}{\partial u_2^2\partial u^2}-71865 u_3u_1
+27u_1^7\frac{\partial b}{\partial u_1}\nonumber\\
&&+16u_1^{12}\left(\frac{\partial b}{\partial u_2}\right)^2
+756u_1^8u_2u_3\frac{\partial^4b}{\partial u_2^3\partial u_1}+392590u_2^2+2268u_1^7u_2u_3
\frac{\partial^3b}{\partial u_2^3}+378u_1^8u_2^2\frac{\partial^4b}{\partial u_2^2\partial u_1^2}\nonumber\\
&&+36bu_1^{12}\frac{\partial^2b}{\partial u_2^2}+2868u_1^6u_2^2\frac{\partial^2b}{\partial u_2^2}
+756u_1^9u_3\frac{\partial^4b}{\partial u_2^3\partial u}
+3926u_1^6u_2\frac{\partial b}{\partial u_2}
+756u_1^9u_2\frac{\partial^4b}{\partial u_2^2\partial u_1\partial u}\nonumber\\
&&+378u_1^8u_3^2\frac{\partial^4b}{\partial u_2^4}
+630u_1^8u_3\frac{\partial^3b}{\partial u_2^2\partial u_1}
+252u_1^8u_2\frac{\partial^3b}{\partial u_2\partial u_1^2}
+1929u_1^7u_2\frac{\partial^2b}{\partial u_2\partial u_1}
+2268u_1^7u_2^2\frac{\partial^3b}{\partial u_2^2\partial u_1}\nonumber\\
&&+2646u_1^8u_2\frac{\partial^3b}{\partial u_2^2\partial u}
+2466u_1^7u_3\frac{\partial^2b}{\partial u_2^2}
+378u_1^8u_4\frac{\partial^3b}{\partial u_2^3}.
\end{eqnarray}
It can be shown that the remaining compatibility conditions for the mixed derivatives of $c$ follow from (\ref{b4}), (\ref{b3}), (\ref{b1}) and (\ref{mixed}).
Solving the system  of partial differential equations (\ref{b4}), (\ref{b3}), (\ref{b1}) and (\ref{mixed}) for
the unknown function $b(u,u_1,u_2,u_3,u_4)$ we arrive at the formula
$$b=\frac{1}{2u_1^6}\left(10u_1u_3-55u_2^2+2\alpha(u)u_1^4\right),$$
where $\alpha(u)$ is an arbitrary function. Substituting the above expression for $b$ into the conditions (\ref{c6}), (\ref{c5}), (\ref{c4}), (\ref{c3}), (\ref{c2}) and (\ref{c1}) and solving the
resulting system of partial differential equations for the unknown function $c(u,u_1,u_2,u_3,u_4,u_5,u_6)$ we obtain
\begin{eqnarray*}c&=&\frac{1}{u_1^8}\left(3u_1^6u_2\frac{\partial \alpha(u)}{\partial u}+2u_1^5u_3\alpha(u)
-6u_1^4u_2^2\alpha(u)+\beta(u)u_1^8- 3u_1^3u_5+ 65u_1^2u_2u_4+ 50u_1^2u_3^2\right.\\
&&\left.- 615u_1u_2^2u_3+ 735u_2^4\right),\end{eqnarray*}
where $\beta(u)$ is another arbitrary function.
\end{proof}

Now let us turn to the property of having momentum.
Notice that the Fr\'{e}chet derivative of the variational  derivative of an arbitrary functional is a self-adjoint  differential operator, see e.g.\ \cite{olver}. The following proposition states that every differential function $h$ whose Fr\'{e}chet derivative is a self-adjoint operator and which satisfies the condition $\mathfrak{D}(h)=u_1$, where $\mathfrak{D}$ is a fifth-order Hamiltonian operator with the leading coefficient of differential order less than or equal to 1, is of the form $h=h(x,u)$. It can be easily verified that any differential function of this form is the variational derivative of the functional $\mathcal{P}=\int\int h(x,u)\mathrm{d}u\mathrm{d}x$. Thus, instead of looking for a functional $\mathcal{P}$ such that
$\mathfrak{D}\delta_u\mathcal{P}=u_1$ it suffices to check the existence of a differential function $h(x,u)$ that satisfies the condition $\mathfrak{D}(h)=u_1.$
\begin{prp}
Let $\mathfrak{D}$ be a fifth-order Hamiltonian operator whose leading coefficient is of differential order less than or equal to 1,
$$\mathfrak{D}=aD_x^5+D_x^5\circ a+bD_x^3+D_x^3\circ b+cD_x+D_x\circ c,\ \mathrm{ord}(a)\leq 1.$$
If there is a differential function $h[u]$ such that $\mathfrak{D}(h)=u_x$ and
$\mathrm{D}_h=(\mathrm{D}_h)^*$, then $h=h(x,u)$.
\end{prp}
\begin{proof}
As it was already mentioned above, the highest possible value $m$ of the level of a
fifth-order Hamiltonian operator is $m=7$,
so we have  $\mathrm{ord}(a)\leq 1$, $\mathrm{ord}(b)\leq 4$ and $\mathrm{ord}(c)\leq 6$.  

Using the relations for the coefficients $a,b,c$ and their derivatives from \cite{cooke} 
we see that $\frac{\partial c}{\partial u_6 }=0$ and $\frac{\partial b}{\partial u_4}=0$, 
so $\mathrm{ord}(c)\leq 5$ and $\mathrm{ord}(b)\leq 3$.
Suppose that $\mathrm{ord}(h)=K\geq 2$.
Then $\mathrm{ord}(D_x^5(h))=5+K$ and differentiating the identity $\mathfrak{D}(h)=u_x$ with respect to $u_{5+K}$ we obtain
$$0=\frac{\partial}{\partial u_{5+K}}(D_x^5(h))=\frac{\partial h}{\partial u_K}.$$
Therefore, $\mathrm{ord}(h)\leq1$. Using the condition $\mathrm{D}_h={(\mathrm{D}_h)}^*$ we get
$\frac{\partial h}{\partial u_1}=0$ as claimed.
\end{proof}
\begin{prp}\label{prp1}
No fifth-order Hamiltonian operator with the leading coefficient $\pm 1$ has momentum.
\end{prp}
\begin{proof}
We have already noticed that it is sufficient to show that no fifth-order Hamiltonian operator
with the leading coefficient $1$ has momentum because the operator with the leading coefficient $-1$
can be transformed into the case under study by a special contact transformation which preserves the property of
 having (or not having) momentum.  Now consider an operator with the
leading coefficient $1$ and suppose that this operator has momentum. Then, differentiating the
condition $\mathfrak{D}(h)=u_1$ with respect to $u_5$, we get
$$\frac{\partial h}{\partial u}=-\frac{h}{2(u+\alpha)},$$ which implies that $h(x,u)=f(x)/\sqrt{u+\alpha(x)}$.
Next, differentiating the condition $\mathfrak{D}(h)=u_1$ with respect to $u_3$ we arrive at
 $\frac{\partial^2 f }{\partial x^2}=-f\gamma$. Finally, differentiating $\mathfrak{D}(h)=u_1$ with
respect to $u_1$ and substituting for $h$ and $\frac{\partial^2 f}{\partial x^2}$, we arrive at
 $0=1$, which is a contradiction. Thus, no fifth-order Hamiltonian operator with the leading coefficient $\pm 1$ has momentum.
\end{proof}
\begin{prp}\label{prp2}
Any fifth-order translation-invariant Hamiltonian operator with the leading coefficient $\pm 1/u_1^4$ has momentum,
and the corresponding functional $\mathcal{P}$ is of the form $\mathcal{P} =\int\int h(u)\mathrm{d}u\mathrm{d}x$,
where $h(u)$ is a solution of the equation
\begin{equation}\label{eqh}\pm\frac{\partial^5 h}{\partial u^5}+ 2\alpha(u)\frac{\partial ^3 h}{\partial u^3}
+3\frac{\partial\alpha(u)}{\partial u}\frac{\partial^2 h}{\partial u^2}
+3\frac{\partial^2\alpha(u)}{\partial u^2}\frac{\partial h}{\partial u}
+2\beta(u)\frac{\partial h}{\partial u}+\frac{\partial^3\alpha(u)}{\partial u^3}h
+\frac{\partial \beta(u)}{\partial u}h-1=0,
\end{equation} 
where $\alpha(u)$
and $\beta(u)$ are as in Lemma \ref{leadcoeff2}.

\end{prp}

\begin{proof}
We will prove our claim only for the case of the leading coefficient  $1/u_1^4$. The proof for the case of the leading coefficient  equal to $-1/u_1^4$ can be obtained in a very similar fashion.

Suppose we are given a fifth-order Hamiltonian operator with the leading coefficient $1/u_1^4$ which has momentum.
Differentiating the condition $\mathfrak{D}(h)=u_1$ with respect
to $u_5$ we obtain $\frac{\partial h}{\partial x}=0$. Substituting this into the condition
$\mathfrak{D}(h)=u_1$ yields (\ref{eqh}).
%
\end{proof}
Combining Propositions \ref{prp1} and \ref{prp2} with the fact that special contact transformations preserve existence of momentum, we arrive at our main result.
\begin{vt}
\begin{enumerate}
\item{A fifth-order Hamiltonian which is not translation-invariant cannot have momentum.}
\item{A fifth-order translation-invariant Hamiltonian operator that can be transformed using a special contact transformation into an operator with the leading coefficient $\pm 1$ cannot have momentum.}
\item{Any fifth-order translation-invariant Hamiltonian operator that can be transformed using a special contact transformation into an operator with the leading coefficient $\pm 1/u_1^4$ has momentum.}
\end{enumerate}
\end{vt}

\section{Examples}
\begin{pr}
 The operator
$$\mathfrak{D}=\frac{1}{2u_1^4 }D_x^5+ D_x^5\circ\frac{1}{2u_1^4 }+bD_x^3+D_x^3\circ b+
cD_x+D_x\circ c,$$
where
\begin{eqnarray*}
b&=&\frac{1}{2u_1^6}\left( 10u_3u_1- 55u_2^2+ u_1^4\right),\\
c&=&\frac{1}{u_1^8}\left(u_1^5u_3
-3u_1^4u_2^2-u_1^8- 3u_1^3u_5+ 65u_1^2u_2u_4+ 50u_1^2u_3^2- 615u_1u_2^2u_3+ 735u_2^4\right),
\end{eqnarray*}
is of the form from Lemma \ref{leadcoeff2} ($\alpha=(1/2)$, $\beta=-1$). The function $h(u)=(-1/2)u$ is a solution of the ordinary differential equation
$$\frac{\partial^5 h}{\partial u^5}+ \frac{\partial ^3 h}{\partial u^3}
-2\frac{\partial h}{\partial u}-1=0. $$
The functional $\mathcal{P}=-(1/4) \int  u^2 \/ \mathrm{d}x$ satisfies the condition $\mathfrak{D}\delta_u\mathcal{P}=u_1$.
\end{pr}
\begin{pr}
 The operator
$$\mathfrak{D}=\frac{1}{2u_1^4 }D_x^5+ D_x^5\circ\frac{1}{2u_1^4 }+bD_x^3+D_x^3\circ b+
cD_x+D_x\circ c,$$
where
\begin{eqnarray*}
b&=&\frac{1}{2u_1^6}\left( 10u_3u_1- 55u_2^2+2\sin(u) u_1^4\right),\\
c&=&\frac{1}{u_1^8}\left(3u_1^6u_2 \cos(u)+2u_1^5u_3\sin(u)
-6u_1^4u_2^2\sin(u)+(\sin(u)+u) u_1^8- 3u_1^3u_5\right.\\
&&\left.+ 65u_1^2u_2u_4+ 50u_1^2u_3^2- 615u_1u_2^2u_3+ 735u_2^4\right),
\end{eqnarray*}
is of the form from Lemma \ref{leadcoeff2} ($\alpha=\sin(u)$, $\beta=\sin(u)+u$). The function $h(u)=1$ is a solution of the ordinary differential equation
$$\frac{\partial^5 h}{\partial u^5}+2\sin(u)\frac{\partial^3 h}{\partial u^3} +3\cos(u)\frac{\partial ^2 h}{\partial u^2}-\sin(u)\frac{\partial  h}{\partial u}
+2\frac{\partial h}{\partial u}+h-1=0,$$
and hence the functional $\mathcal{P}=\int u\ \mathrm{d}x$ satisfies the condition $\mathfrak{D}\delta_u\mathcal{P}=u_1$.
\end{pr}
\section*{Acknowledgements}
The author thanks Dr. A. Sergyeyev for stimulating discussions. This research was supported by the Silesian university in Opava under the student grant SGS/18/2010, by the Ministry of Education, Youth and Sports of the Czech Republic under the grant
MSM 4781305904, and by the fellowship
from the Moravian--Silesian region.\looseness=-1

\end{document}